\documentclass[conference]{IEEEtran}
\IEEEoverridecommandlockouts

\usepackage{mathpazo}
\usepackage{times}
\usepackage{cite}
\usepackage{amsmath,amssymb,amsfonts}
\usepackage{algorithmic}
\usepackage{graphicx}
\usepackage{textcomp}
\usepackage{xcolor}
\usepackage{breqn}
\usepackage{mathtools}
\usepackage{bbm}
\usepackage{bm}
\usepackage{amsthm}
\usepackage{balance}

\DeclareSymbolFont{cmsymbols}{OMS}{cmsy}{m}{n}
\let\emptyset\relax
\DeclareMathSymbol{\emptyset}{\mathord}{cmsymbols}{"3B}

\def\BibTeX{{\rm B\kern-.05em{\sc i\kern-.025em b}\kern-.08em
    T\kern-.1667em\lower.7ex\hbox{E}\kern-.125emX}}

\newtheorem{lemma}{Lemma}
\newtheorem{theorem}{Theorem}

\newtheorem{remark}{Remark}

\begin{document}

\pagestyle{plain}

\title{
Secure and Private Structured-Subset Retrieval: Fundamental Limits and Achievable Schemes
}

\author{
\IEEEauthorblockN{Maha Issa and Anoosheh Heidarzadeh}
\IEEEauthorblockA{
Department of Electrical and Computer Engineering\\
Santa Clara University, Santa Clara, CA, USA\\
\{missa,aheidarzadeh\}@scu.edu
}
}

\maketitle

\begin{abstract}
This work introduces the \emph{Secure and Private Structured-Subset Retrieval (SPSSR)} problem. In SPSSR, a user wishes to retrieve one subset from an arbitrary family of size-$D$ subsets from $K$ messages replicated across $N$ non-colluding servers that share randomness unknown to the user. The privacy requirement ensures that no server learns which subset is requested, while the security requirement ensures that the user learns nothing about the messages outside the requested subset. This generalizes Symmetric Multi-message Private Information Retrieval (SMPIR), where the candidate demand sets consist of all size-$D$ subsets. We show that, for every candidate demand family, the maximum achievable retrieval rate is equal to ${1-1/N}$. We also show that the minimum ratio between the size of the shared randomness and the message size required to achieve this rate is ${D/(N-1)}$, and that, for balanced linear SPSSR schemes, the minimum required subpacketization level is ${(N-1)/\gcd(D,N-1)}$; both quantities are independent of the demand family. Our converse proof for the maximum achievable retrieval rate applies to arbitrary demand families, unlike the existing proof for SMPIR, which is tailored to the full demand family. For achievability, we construct a single SPSSR scheme that applies uniformly to every demand family, achieves the optimal retrieval rate with the optimal shared-randomness ratio, and requires the optimal subpacketization level among balanced linear schemes. This subpacketization level is no larger than that of known SMPIR schemes in any parameter regime and is smaller in some regimes.
\end{abstract}

\thispagestyle{plain}

\section{Introduction}

In the Private Information Retrieval (PIR) problem, a user wishes to retrieve one message from a database stored on one or multiple servers without revealing the identity of the desired message to any server~\cite{SJ2017,SJ2016ArbitraryTIFS,SJ2018Multiround,TSC2019,VBU2022}.
The retrieval efficiency of a PIR scheme is measured by its retrieval rate, defined as the ratio of the amount of desired data to the total amount of retrieved data. 
The maximum achievable retrieval rate for arbitrary numbers of servers $N$ and messages $K$ was characterized in~\cite{SJ2017}.

Several variations of the PIR problem have been studied in the literature; see, e.g.,~\cite{VWU2023,UAGJTT2022} and the references therein.
One important variation is Symmetric PIR (SPIR)~\cite{SJ2019Sym,WU2021}, which strengthens the classical PIR formulation by imposing a security requirement in addition to privacy.
Specifically, the user must not learn any information about messages other than the desired one.
It was shown in~\cite{SJ2019Sym} that the maximum achievable retrieval rate is ${1-1/N}$, provided that the servers have access to sufficient shared randomness that is unknown to the user.

Another important extension is the Multi-message PIR (MPIR) problem~\cite{BU2018,WHS2022,HWS2025,WHS2025,HZTS2025}, where the user wishes to retrieve $D\geq 2$ messages simultaneously. 
Several MPIR schemes have been developed, including those in~\cite{BU2018,HWS2025}, achieving optimal or near-optimal retrieval rates in different parameter regimes.
The secure counterpart of MPIR, referred to as Symmetric MPIR (SMPIR), was studied in~\cite{WBU2022}. 
In SMPIR, the servers require that the user learn no information about messages outside the demand set, while the privacy requirement ensures that no server can identify which demand subset is requested. 
Similarly to SPIR, it was shown in~\cite{WBU2022} that the maximum achievable retrieval rate is equal to ${1-1/N}$, provided that the servers share sufficient amount of randomness unknown to the user.

The MPIR and SMPIR formulations assume that every subset of a given size can potentially be requested by the user. 
However, in many practical applications, the servers may already know that only certain subsets are feasible demands. 
For example, in a genomic database, a researcher may be interested only in retrieving groups of genes associated with specific biological pathways, rather than arbitrary collections of genes. 
Similarly, in a recommendation system, a user may be interested only in predefined item collections, such as genre-based or theme-based collections, rather than arbitrary subsets of the entire catalog.
In such scenarios, many subsets of messages can never appear as valid demands. 
More generally, the set of candidate demand subsets may be restricted by the structure imposed by the underlying application.

Motivated by such scenarios, our parallel work~\cite{IH2026PSSRarXiv} introduces the Private Structured-Subset Retrieval (PSSR) problem, which generalizes MPIR by replacing the full family of all $D$-subsets of messages with an arbitrary family of candidate demand subsets, each of size $D$.
Although any MPIR scheme remains valid for the PSSR setting, it is shown in~\cite{IH2026PSSRarXiv} that exploiting the demand structure can improve retrieval efficiency. 
In particular, it can lead to higher achievable retrieval rates and a smaller subpacketization level, compared to schemes designed for the full demand family.

These observations raise the following question: can demand structure be exploited to improve retrieval efficiency when security is required in addition to privacy?
In particular, can it lead to higher retrieval rates than SMPIR schemes, or reduce the subpacketization level or the ratio between the size of the shared randomness and the message size?

This work introduces the \emph{Secure and Private Structured-Subset Retrieval (SPSSR)} problem, which extends SMPIR to settings with structured demand families. 
In SPSSR, the user wishes to retrieve one demand subset from a known family of candidate subsets, each of size $D$, while preventing the servers from identifying which subset is requested. 
At the same time, the security requirement ensures that the user obtains no information about messages outside the demand subset. 

We restrict attention to \emph{balanced ${\{0,1\}}$-linear SPSSR schemes}, in which all messages are divided into the same number of equal-size subpackets, the user retrieves the same number of linear combinations from each server, and each combination is formed using coefficients in ${\{0,1\}}$ from message subpackets and shared-randomness symbols of the same size as a message subpacket.
This class is attractive because it evenly distributes communication and computation, works over arbitrary finite fields, admits low-complexity encoding and decoding operations, and includes the best-known achievable schemes for SPIR~\cite{SJ2019Sym} and SMPIR~\cite{WBU2022}.

For this class of schemes, we characterize the maximum achievable retrieval rate as ${1-1/N}$. 
We also show that the minimum shared-randomness ratio and subpacketization level required to achieve this rate are ${D/(N-1)}$ and ${(N-1)/\gcd(D,N-1)}$, respectively.
These results are independent of the demand family and coincide with their SMPIR counterparts for the full demand family. 
Moreover, the retrieval-rate and shared-randomness-ratio results hold for arbitrary SPSSR schemes, whereas the subpacketization-level result holds for balanced linear SPSSR schemes.

Our converse proof for the maximum achievable retrieval rate applies to
arbitrary demand families, unlike the proof of~\cite{WBU2022} for SMPIR, which is tailored to the full demand family.
For achievability, we construct a single SPSSR scheme that applies uniformly to every demand family, achieves the optimal retrieval rate with the optimal shared-randomness ratio, and requires the optimal subpacketization level among all balanced linear schemes.
The scheme generalizes the SPIR scheme of~\cite{SJ2019Sym} from retrieving one message to retrieving an entire demand set simultaneously. 
This avoids running a separate SPIR scheme for each message in the demand set, which is the best previously known rate-optimal SMPIR approach from the perspective of subpacketization, and can in turn reduce the required subpacketization level.

\section{Problem Setup}

For any integers ${i,j}$ such that ${0\leq i\leq j}$, we denote the set ${\{i,i+1,\dots,j\}}$ by ${[i:j]}$.
We denote random variables by bold-face symbols and their realizations by regular symbols.
We fix an arbitrary prime power $q$ throughout, denote the finite field of order $q$ by $\mathbb{F}_q$, and denote the $L$-dimensional vector space over $\mathbb{F}_q$ by $\mathbb{F}_q^L$ for any integer $L\geq 1$.
All entropy and mutual information quantities are measured in $q$-ary units.

Consider a user that interacts with ${N\geq 2}$ non-colluding servers.
Each server stores the same set of $K$ messages ${\mathrm{X}_1,\dots,\mathrm{X}_K}$ and has access to $M$ secret keys $\mathrm{S}_1,\dots,\mathrm{S}_M$ that are shared
among the servers and unknown to the user.
Each message $\mathrm{X}_i$, ${i \in [1:K]}$, consists of $L$ symbols from ${\mathbb{F}_q}$, i.e., ${\mathrm{X}_i\in \mathbb{F}_{q}^{L}}$, and each shared secret key $\mathrm{S}_m$, ${m\in [1:M]}$, is a symbol from $\mathbb{F}_q$, i.e., ${\mathrm{S}_m\in \mathbb{F}_q}$. 
For any ${\mathrm{U}\subseteq[1:K]}$, define ${\overline{\mathrm{U}}\coloneqq[1:K]\setminus \mathrm{U}}$, ${\mathrm{X}_{\mathrm{U}}\coloneqq\{\mathrm{X}_i: i\in \mathrm{U}\}}$, and ${\mathrm{X}_{\overline{\mathrm{U}}}\coloneqq \{\mathrm{X}_i: i\in \overline{\mathrm{U}}\}}$.
 
The user wishes to retrieve $D$ messages, for some ${2\leq D\leq K-1}$, indexed by $\mathrm{W}\in \{\mathrm{W}_1,\dots,\mathrm{W}_E\}$, where each $\mathrm{W}_j$ is a subset of $[1:K]$ of size $D$. 

We refer to $\mathrm{X}_{\mathrm{W}}$ as the \emph{demand messages},  ${\mathrm{X}_{\overline{\mathrm{W}}}}$ as the \emph{interference messages}, $\mathrm{W}$ as the \emph{demand index set}, and ${\mathrm{W}_1,\dots,\mathrm{W}_E}$ as the \emph{candidate demand index sets}.  

In this work, we assume the following:
\begin{itemize}
    \item The random variables ${\mathbf{X}_1,\dots,\mathbf{X}_K}$ are independent and uniformly distributed over ${\mathbb{F}_{q}^{L}}$, which implies that ${H(\mathbf{X}_1,\dots,\mathbf{X}_K)=KL}$, and more generally, ${H(\mathbf{X}_{\mathrm{U}})=|\mathrm{U}|L}$ for any ${\mathrm{U}\subseteq[1:K]}$. 
    \item The random variables ${\mathbf{S}_1,\dots,\mathbf{S}_M}$ are independent and uniformly distributed over ${\mathbb{F}_{q}}$, which implies that ${H(\mathbf{S}_1,\dots,\mathbf{S}_M)=M}$.
    \item The random variable $\mathbf{W}$ is distributed arbitrarily over $\{\mathrm{W}_1,\dots,\mathrm{W}_E\}$, subject to the condition that every $\mathrm{W}_j$, $j \in [1:E]$, has a nonzero probability.
    \item $\mathbf{X}_{[1:K]}$, $\mathbf{S}\coloneqq (\mathbf{S}_1,\dots,\mathbf{S}_M)$, and $\mathbf{W}$ are independent random variables.
\end{itemize}

The user generates $N$ queries ${\mathrm{Q}^{[\mathrm{W}]}_{1},\dots,\mathrm{Q}^{[\mathrm{W}]}_{N}}$, and sends query ${\mathrm{Q}^{[\mathrm{W}]}
_{n}}$ to server $n$ for each $n\in [1:N]$. 
Each query is a (possibly stochastic) function of the demand index set, generated without a prior access to the messages or the shared secret keys, i.e., 
\begin{equation}
\label{eq:Q_indep_X_S}
I(\mathbf{Q}_{[1:N]}^{[\mathrm{W}]};\mathbf{X}_{[1:K]},\mathbf{S})=0,
\end{equation} 
where $\mathbf{Q}_{[1:N]}^{[\mathrm{W}]}\coloneqq \{\mathbf{Q}^{[\mathrm{W}]}_{1},\dots,\mathbf{Q}^{[\mathrm{W}]}_{N}\}$. 

Upon receiving the query ${\mathrm{Q}_{n}^{[\mathrm{W}]}}$, each server $n$ computes an answer ${\mathrm{A}_{n}^{[\mathrm{W}]}}$ and sends it back to the user. 
Answers are deterministic functions of the queries, messages, and shared secret keys, i.e., 
\begin{equation} 
\label{eq:A_func_Q_X_S}
H(\mathbf{A}_{n}^{[\mathrm{W}]} | \mathbf{Q}_{n}^{[\mathrm{W}]}, \mathbf{X}_{[1:K]},\mathbf{S})=0, \quad \forall n \in [1:N].
\end{equation}

Upon receiving all the servers' answers, the user must be able to recover the demand messages, i.e.,
\begin{equation} 
\label{eq:correctness}
H( \mathbf{X}_{\mathrm{W}} | \mathbf{Q}_{[1:N]}^{[\mathrm{W}]}, \mathbf{A}_{[1:N]}^{[\mathrm{W}]})=0,
\end{equation} where $\mathbf{A}_{[1:N]}^{[\mathrm{W}]}\coloneqq \{\mathbf{A}^{[\mathrm{W}]}_{1},\dots,\mathbf{A}^{[\mathrm{W}]}_{N}\}$. 
We refer to this requirement as the \emph{correctness condition}. 

The information available to any server must reveal no information regarding the realization $\mathrm{W}$, i.e.,
\begin{equation}
\label{eq:privacy}	I(\mathbf{W};\mathbf{Q}_{n}^{[\mathrm{W}]},\mathbf{A}_{n}^{[\mathrm{W}]},\mathbf{X}_{[1:K]},\mathbf{S})=0, \quad \forall n \in [1:N].
\end{equation}
This requirement, which we refer to as the \emph{privacy condition}, keeps the user's demand index set private from any server.

Moreover, the servers wish to prevent the user from obtaining any information about the interference messages, i.e.,  
\begin{equation}
\label{eq:security}	
I(\mathbf{X}_{\overline{\mathrm{W}}};\mathbf{Q}_{[1:N]}^{[\mathrm{W}]},\mathbf{A}_{[1:N]}^{[\mathrm{W}]})=0.
\end{equation}
We refer to this requirement as the \emph{security condition}.

The problem is to design a scheme that satisfies the correctness, privacy, and security conditions. 
We refer to this problem as \emph{Secure and Private Structured-Subset Retrieval (SPSSR)}. 

In this work, we assume, without loss of generality, that $\big|\bigcup_{j=1}^{E}\mathrm{W}_j\big|=K$ and
$\big|\bigcap_{j=1}^{E}\mathrm{W}_j\big|=0$.
Indeed, if an index ${i\in[1:K]}$ appears in none of the
$\mathrm{W}_j$'s, then the message $\mathrm{X}_i$ can be removed, yielding an equivalent instance with $K-1$ messages. 
If instead $i$ appears in all of the $\mathrm{W}_j$'s, then $\mathrm{X}_i$ can be retrieved directly---while still meeting the above requirements---and removing $i$ from each $\mathrm{W}_j$ again yields an instance with $K-1$ messages.

When ${D\mid K}$ and ${E = K/D}$, the $E$ candidate demand subsets, each of size ${D=K/E}$, form a partition of the $K$ messages.  
In this case, each candidate demand subset is perceived as a super-message consisting of $K/E$ messages.
Thus, the SPSSR problem reduces to the SPIR problem~\cite{SJ2019Sym}, where the user privately and securely retrieves one of the $E$ super-messages stored on the $N$ servers.
When ${E=\binom{K}{D}}$, the SPSSR problem is equivalent to the SMPIR problem~\cite{WBU2022}, in which every $D$-subset of the messages is a candidate demand. 

In this work, we focus on a class of SPSSR schemes that we refer to as \emph{balanced $\{0,1\}$-linear SPSSR schemes}.
In such schemes, each message is partitioned into $L$ subpackets, each consisting of a single message symbol. 
The user queries each server for a collection of linear combinations of message subpackets and shared secret keys with coefficients in $\{0,1\}$, and
each server answers the user with the corresponding linear combinations.
Moreover, the queries sent to all servers have the same total length, and the answers returned by all servers have the same total length.

We define the \emph{retrieval rate} of a balanced ${\{0,1\}}$-linear SPSSR scheme as the ratio between the amount of information required by the user and the total amount of information retrieved from all servers, namely
\begin{equation}\label{eq:RateDef}
\frac{H(\mathbf{X}_{\mathbf{W}})}{\sum_{n=1}^N H (\mathbf{A}_{n}^{[\mathbf{W}]} | \mathbf{Q}_{n}^{[\mathbf{W}]})}.
\end{equation}

We refer to the number of subpackets per message, $L$, as the \emph{subpacketization level}, and to the number of shared secret keys normalized by the subpacketization level, ${M/L}$, as the \emph{shared-randomness ratio}.

The goals of this work are threefold:
\begin{itemize}
\item[(i)] to characterize the maximum retrieval rate achievable by balanced
${\{0,1\}}$-linear SPSSR schemes, over all subpacketization levels and
shared-randomness ratios, in terms of the number of servers $N$ and the
candidate demand index sets $\mathrm{W}_1,\dots,\mathrm{W}_E$ (and, in turn, the total number of messages $K$ and the number of demand messages $D$); 
\item[(ii)] to characterize the
minimum shared-randomness ratio and the minimum subpacketization level required to achieve this maximum rate; and
\item[(iii)] to determine whether these two minima can be achieved simultaneously.
\end{itemize}

\section{Main results}

In this section, we present our converse and achievability results for the SPSSR problem.

\begin{theorem}
\label{thm:SPSSR_capacity}
For $N$ servers and $E$ candidate demand index sets $\mathrm{W}_1,\dots,\mathrm{W}_E$, the maximum rate achievable by any balanced ${\{0,1\}}$-linear SPSSR scheme is
\begin{equation}
\label{eq:SPSSR_cap}
{R} \coloneqq 1-\frac{1}{N},
\end{equation}
independently of the particular choice of
$\mathrm{W}_1,\dots,\mathrm{W}_E$.
\end{theorem}

The converse proof is presented in Section~\ref{sec:rate_conv_proof_SPSSR}.
The result holds for all SPSSR schemes and therefore applies, in particular, to the balanced ${\{0,1\}}$-linear SPSSR schemes considered here.
The proof differs from that in~\cite{WBU2022} for the SMPIR setting.
Their proof relies on two key technical lemmas, \cite[Lemmas~3 and~4]{WBU2022}, which are tailored to the full-demand-family setting and do not extend to arbitrary demand families. 
In contrast, we introduce a new lemma that applies to any demand family. 
As a byproduct, our proof also yields an alternative converse proof for the SMPIR setting.

To prove achievability, we construct a balanced ${\{0,1\}}$-linear SPSSR scheme whose rate matches the converse bound. 
The scheme, presented in
Section~\ref{sec:achiev_proof_SPSSR}, extends the SPIR scheme of~\cite{SJ2019Sym} to the multi-demand setting. 
The main idea is to use shared random linear combinations and a common masking structure for the entire demand set, rather than running a separate SPIR scheme for each demand message.

\begin{theorem}
\label{thm:SPSSR_L*_bounds}
For $N$ servers and $E$ candidate demand index sets
$\mathrm{W}_1,\dots,\mathrm{W}_E$, the minimum shared-randomness ratio and the minimum subpacketization level required by any balanced ${\{0,1\}}$-linear SPSSR scheme achieving rate $R$ are respectively
\begin{equation}
    \label{eq:kappa_SPSSR}
    \frac{M}{L} = \frac{D}{N-1}
\end{equation}
and 
\begin{equation}
    \label{eq:L*_SPSSR}
    L = \frac{N-1}{\gcd(D,N-1)},
\end{equation}
both independently of the particular choice of
$\mathrm{W}_1,\dots,\mathrm{W}_E$, and these two minima can be achieved simultaneously.
\end{theorem}

The converse proof for~\eqref{eq:kappa_SPSSR} is presented in Section~\ref{sec:M/L_conv_proof_SPSSR} and follows the approach in~\cite{WBU2022} for the SMPIR setting.
The converse proof for~\eqref{eq:L*_SPSSR} is presented in Section~\ref{sec:L_conv_proof_SPSSR} and follows from the balanced structure of the scheme.
The achievability of both \eqref{eq:kappa_SPSSR} and~\eqref{eq:L*_SPSSR}
is established by the proposed scheme.

\begin{remark}
\normalfont 
Interestingly, in the SPSSR setting, restricting demands to a structured family does not change the optimal rate. Specifically, the optimal rate coincides with that of the SMPIR setting~\cite{WBU2022}, which pertains to the full demand family, as well as the SPIR setting~\cite{SJ2019Sym}. 
This is in contrast to the PSSR setting~\cite{IH2026PSSRarXiv}, where the demand family is restricted in a similar way and the privacy requirement is the same, but there is no security requirement. 
In particular, as shown in~\cite{IH2026PSSRarXiv}, the optimal rate for PSSR can be strictly larger than the optimal rate in the MPIR setting, which pertains to the full demand family.
Therefore, when security is imposed in addition to privacy, exploiting structure in the candidate demands offers no advantage with respect to rate.
Thus, and since any SMPIR scheme also applies when the demand space is restricted to an arbitrary subfamily, any rate-optimal SMPIR scheme is also rate-optimal for the SPSSR setting. 
This includes the SMPIR scheme of~\cite{WBU2022}, as well as the scheme obtained by applying the SPIR scheme of~\cite{SJ2019Sym} to retrieve the $D$ demand messages one at a time.
\end{remark}

\begin{remark}
\normalfont
The minimum shared-randomness ratio and subpacketization level required to achieve the optimal rate are independent of the demand family in the SPSSR setting and coincide with those in the SMPIR setting. 
Thus, exploiting structure in the candidate demands provides no further gain in these two metrics. 
In particular, any SMPIR scheme that is optimal in shared-randomness ratio and subpacketization level remains optimal for SPSSR.
Existing SMPIR schemes achieve the optimal shared-randomness ratio, but they are not necessarily optimal in subpacketization. 
In particular, the scheme of~\cite{WBU2022} requires subpacketization level at least ${N^{K-D+1}/D}$, which grows exponentially with the number of messages $K$; see~\cite{H2026} for details. 
Alternatively, applying the SPIR scheme of~\cite{SJ2019Sym} independently $D$ times also achieves the optimal shared-randomness ratio, but requires subpacketization level ${N-1}$. 
This is independent of $K$, but does not decrease with $D$.
In contrast, our SPSSR scheme achieves the optimal rate with the optimal shared-randomness ratio, while requiring subpacketization level at most ${N-1}$, and strictly less than ${N-1}$ whenever $D$ and ${N-1}$ are not coprime.
\end{remark}

\section{Converse Proofs}

This section presents the converse proofs for
Theorems~\ref{thm:SPSSR_capacity} and~\ref{thm:SPSSR_L*_bounds}: the upper
bound $R$ in~\eqref{eq:SPSSR_cap} on the maximum achievable retrieval rate, and
the lower bounds $M/L$ and $L$ in~\eqref{eq:kappa_SPSSR}
and~\eqref{eq:L*_SPSSR} on the shared-randomness ratio and subpacketization
level required to achieve rate $R$, respectively.

\subsection{Upper Bounding the Achievable Rate}
\label{sec:rate_conv_proof_SPSSR}

The proof relies on the following three lemmas. 
The first two generalize \cite[Lemmas~1 and 2]{WBU2022}, which correspond to the special case $\mathrm{U}=\mathrm{W}$.

\begin{lemma}
    \label{lemma: symmetry}
    For any $\mathrm{W},\mathrm{W}'\in \{\mathrm{W}_1,\dots,\mathrm{W}_E\}$, any ${\mathrm{U}\subseteq[1:K]}$, and any $n\in[1:N]$, it holds that
    \begin{align}
         H(\mathbf{A}_{n}^{[\mathrm{W}]} | \mathbf{X}_{\mathrm{U}}, \mathbf{Q}_{n}^{[\mathrm{W}]}) & = H(\mathbf{A}_{n}^{[\mathrm{W}']} | \mathbf{X}_{\mathrm{U}}, \mathbf{Q}_{n}^{[\mathrm{W}']}).  \label{eq: symmetry_lemma_1}
    \end{align}
\end{lemma}
\begin{proof}
This follows directly from the privacy condition. A complete proof can be found in~\cite[Lemma~1]{IH2026PSSRarXiv}.
\end{proof}

\begin{lemma}
\label{lemma: given_Q1:N_=_given_Qn}
For any $\mathrm{W} \in \{\mathrm{W}_1,\dots,\mathrm{W}_E\}$, any $\mathrm{U} \subseteq [1:K]$, and any $n \in [1:N]$, it holds that
\begin{equation}\label{eq:lemma2}
    H(\mathbf{A}_{n}^{[\mathrm{W}]} | \mathbf{Q}_{[1:N]}^{[\mathrm{W}]}, \mathbf{X}_{\mathrm{U}}) = H(\mathbf{A}_{n}^{[\mathrm{W}]} | \mathbf{Q}_{n}^{[\mathrm{W}]}, \mathbf{X}_{\mathrm{U}}).
\end{equation}
\end{lemma}
\begin{proof}
The proof follows the same line of argument as~\cite[Lemma~2]{IH2026PSSRarXiv}. 
The only modification is that, throughout the steps in~\cite[Eqs.~(25)--(32)]{IH2026PSSRarXiv}, the random variable $\mathbf{S}$ is included as an additional argument wherever $\mathbf{X}_{[1:K]}$ appears.
\end{proof}

We now introduce the third lemma, which follows from the security condition.
This lemma has the same flavor as~\cite[Lemma~5]{WBU2022}, but differs from it in two main aspects: its statement and its proof.

The statement of our lemma generalizes that of~\cite[Lemma~5]{WBU2022}, which corresponds to the special case ${|\mathrm{U}|=D}$ and ${\mathrm{U}\neq\mathrm{W}}$.
Moreover, our proof argument applies to any demand family, whereas the proof in~\cite{WBU2022} relies on two additional lemmas tailored to the full demand family and therefore does not extend to the SPSSR setting.

\begin{lemma}
\label{lemma: dropping_undesired_subset_from_condition}
 For any $\mathrm{W} \in \{\mathrm{W}_1,\dots,\mathrm{W}_E\}$, any $\mathrm{U} \subseteq [1:K]$, and any $n \in [1:N]$, it holds that 
    \begin{equation}
    \label{eq:lemma3}
        H( \mathbf{A}_{n}^{[\mathrm{W}]} | \mathbf{X}_{\mathrm{U}}, \mathbf{Q}_{n}^{[\mathrm{W}]}  ) = H (\mathbf{A}_{n}^{[\mathrm{W}]} |   \mathbf{Q}_{n}^{[\mathrm{W}]} ).
    \end{equation}
\end{lemma}

\begin{proof}
  We have 
\begin{align}
         & I(\mathbf{A}_{n}^{[\mathrm{W}]}; \mathbf{X}_{\mathrm{U}} |  \mathbf{Q}_{n}^{[\mathrm{W}]}) \nonumber
         \\ & \quad = H(\mathbf{A}_{n}^{[\mathrm{W}]} |   \mathbf{Q}_{n}^{[\mathrm{W}]} )  - H(\mathbf{A}_{n}^{[\mathrm{W}]} | \mathbf{X}_{\mathrm{U}},  \mathbf{Q}_{n}^{[\mathrm{W}]} ) \nonumber
         \\ & \quad = H( \mathbf{X}_{\mathrm{U}} | \mathbf{Q}_{n}^{[\mathrm{W}]}) - H( \mathbf{X}_{\mathrm{U}} | \mathbf{Q}_{n}^{[\mathrm{W}]}, \mathbf{A}_{n}^{[\mathrm{W}]}) \nonumber
         \\ & \quad = H( \mathbf{X}_{\mathrm{U}} ) - H( \mathbf{X}_{\mathrm{U}} | \mathbf{Q}_{n}^{[\mathrm{W}]}, \mathbf{A}_{n}^{[\mathrm{W}]}) \label{eq:H_Xw_minus_H_Xw_given_Qn_An_W'}
    \end{align}
where~\eqref{eq:H_Xw_minus_H_Xw_given_Qn_An_W'} follows from~\eqref{eq:Q_indep_X_S}.
To prove~\eqref{eq:lemma3}, it thus suffices to show that
\[
H ( \mathbf{X}_{\mathrm{U}} ) = H ( \mathbf{X}_{\mathrm{U}} | \mathbf{Q}_{n}^{[\mathrm{W}]}, \mathbf{A}_{n}^{[\mathrm{W}]} ).
\]

Let ${\mathrm{V}\coloneqq\mathrm{U} \cap \mathrm{W}}$ and  ${\mathrm{V}'\coloneqq\mathrm{U} \setminus \mathrm{W}}$.
Then,
\begin{align}
         & H( \mathbf{X}_{\mathrm{U}} | \mathbf{Q}_{n}^{[\mathrm{W}]}, \mathbf{A}_{n}^{[\mathrm{W}]}) \nonumber 
         \\ & \quad = H( \mathbf{X}_{\mathrm{V}} , \mathbf{X}_{\mathrm{V}'} | \mathbf{Q}_{n}^{[\mathrm{W}]}, \mathbf{A}_{n}^{[\mathrm{W}]}) \nonumber
         \\ & \quad = H( \mathbf{X}_{\mathrm{V}} | \mathbf{Q}_{n}^{[\mathrm{W}]}, \mathbf{A}_{n}^{[\mathrm{W}]}) + H(\mathbf{X}_{\mathrm{V}'} | \mathbf{X}_{\mathrm{V}}, \mathbf{Q}_{n}^{[\mathrm{W}]}, \mathbf{A}_{n}^{[\mathrm{W}]}). \label{eq:Xw_=_Xj_plus_Xj1_given_Xj}
    \end{align}

We first lower bound
$H(\mathbf{X}_{\mathrm{V}'} |
\mathbf{X}_{\mathrm{V}},\mathbf{Q}_{n}^{[\mathrm{W}]},
\mathbf{A}_{n}^{[\mathrm{W}]})$
as follows: 
\begin{align}
        & H(\mathbf{X}_{\mathrm{V}'} | \mathbf{X}_{\mathrm{V}}, \mathbf{Q}_{n}^{[\mathrm{W}]}, \mathbf{A}_{n}^{[\mathrm{W}]}) \nonumber 
    \\ & \quad \geq  H(\mathbf{X}_{\mathrm{V}'} | \mathbf{X}_{\mathrm{W}}, \mathbf{Q}_{[1:N]}^{[\mathrm{W}]}, \mathbf{A}_{[1:N]}^{[\mathrm{W}]}) \label{eq:X_J1_given_X_W2_Q_A}
    \\ & \quad = H(\mathbf{X}_{\mathrm{V}'} | \mathbf{X}_{\mathrm{W}}, \mathbf{Q}_{[1:N]}^{[\mathrm{W}]}, \mathbf{A}_{[1:N]}^{[\mathrm{W}]}) \nonumber \\ 
    & \quad \quad + H ( \mathbf{X}_{\mathrm{W}} | \mathbf{Q}_{[1:N]}^{[\mathrm{W}]}, \mathbf{A}_{[1:N]}^{[\mathrm{W}]}) \label{eq:recoverability_for_X_w}
    \\ & \quad = H(\mathbf{X}_{\mathrm{V}'}, \mathbf{X}_{\mathrm{W}}| \mathbf{Q}_{[1:N]}^{[\mathrm{W}]}, \mathbf{A}_{[1:N]}^{[\mathrm{W}]}) \label{eq:X_J1_X_w_given_Q_A}
    \\ & \quad = H(\mathbf{X}_{\mathrm{V}'} |\mathbf{Q}_{[1:N]}^{[\mathrm{W}]}, \mathbf{A}_{[1:N]}^{[\mathrm{W}]}) \nonumber \\ 
    & \quad \quad + H( \mathbf{X}_{\mathrm{W}}| \mathbf{X}_{\mathrm{V}'}, \mathbf{Q}_{[1:N]}^{[\mathrm{W}]}, \mathbf{A}_{[1:N]}^{[\mathrm{W}]}) \label{eq:X_J1_given_Q_A_plus_X_w_given_X_J1_Q_A}
    \\ & \quad = H(\mathbf{X}_{\mathrm{V}'}) \label{eq:X_J1}
    \end{align}
where~\eqref{eq:X_J1_given_X_W2_Q_A} follows from the fact that conditioning does not increase entropy (and since ${\mathrm{V} \subseteq \mathrm{W}}$);~\eqref{eq:recoverability_for_X_w} follows from~\eqref{eq:correctness};~\eqref{eq:X_J1_X_w_given_Q_A} and~\eqref{eq:X_J1_given_Q_A_plus_X_w_given_X_J1_Q_A} follow from the chain rule; and~\eqref{eq:X_J1} follows from both~\eqref{eq:correctness} and~\eqref{eq:security}.
Obviously, the reverse inequality, 
\[
H(\mathbf{X}_{\mathrm{V}'} | \mathbf{X}_{\mathrm{V}}, \mathbf{Q}_{n}^{[\mathrm{W}]}, \mathbf{A}_{n}^{[\mathrm{W}]}) \leq H(\mathbf{X}_{\mathrm{V}'}),
\]
also holds, and thus we arrive at 
\begin{equation}
\label{eq:Xj1_given_Xj_QA=_Xj1}
    H(\mathbf{X}_{\mathrm{V}'} | \mathbf{X}_{\mathrm{V}}, \mathbf{Q}_{n}^{[\mathrm{W}]}, \mathbf{A}_{n}^{[\mathrm{W}]})=  H (\mathbf{X}_{\mathrm{V}'} ).
\end{equation}

Next, we consider the term $H ( \mathbf{X}_{\mathrm{V}} | \mathbf{Q}_{n}^{[\mathrm{W}]}, \mathbf{A}_{n}^{[\mathrm{W}]} )$ and show that $H ( \mathbf{X}_{\mathrm{V}} | \mathbf{Q}_{n}^{[\mathrm{W}]}, \mathbf{A}_{n}^{[\mathrm{W}]} ) =  H (\mathbf{X}_{\mathrm{V}} )$ for any ${\mathrm{V} \subseteq \mathrm{W}}$.
We prove this by induction on $|\mathrm{V}|$. 

For the base case, if $|\mathrm{V}|=0$, then $H( \mathbf{X}_{\emptyset} | \mathbf{Q}_{n}^{[\mathrm{W}]}, \mathbf{A}_{n}^{[\mathrm{W}]}) =  H (\mathbf{X}_{\emptyset})=0$.
Now, suppose that, for some $k\geq 0$, the statement holds for every $\mathrm{V} \subseteq \mathrm{W}$ such that $|\mathrm{V}| \leq k$, i.e.,
\begin{equation}
\label{eq:induction_statement}
    H ( \mathbf{X}_{\mathrm{V}} | \mathbf{Q}_{n}^{[\mathrm{W}]}, \mathbf{A}_{n}^{[\mathrm{W}]} ) =  H (\mathbf{X}_{\mathrm{V}} ), \quad \forall \mathrm{V} \subseteq \mathrm{W}, |\mathrm{V}| \leq k.
\end{equation}

We next prove the statement for any
$\mathrm{V} \subseteq \mathrm{W}$ such that ${|\mathrm{V}|=k+1}$. 
Fix an arbitrary ${l\in \mathrm{V}}$.
There exists a demand subset ${\mathrm{W}'\neq \mathrm{W}}$ such that ${l \notin \mathrm{W}'}$. Indeed, otherwise, $l$ would belong to every candidate demand index set, contradicting the assumption that ${\big|\bigcap_{j=1}^{E}\mathrm{W}_j\big|=0}$.

Let $\mathrm{V}_1\coloneqq\mathrm{V}\cap \mathrm{W}'$ and $\mathrm{V}_2\coloneqq\mathrm{V}\setminus (\mathrm{W}'\cup \{l\})$. 
Note that ${|\mathrm{V}_1|\leq k}$, since ${\mathrm{V}_1\subset \mathrm{V}}$, ${|\mathrm{V}|=k+1}$, and ${l\in \mathrm{V}\setminus \mathrm{V}_1}$.
Then, 
\begin{align}
& H (\mathbf{X}_{\mathrm{V}} | \mathbf{Q}_{n}^{[\mathrm{W}]},\mathbf{A}_{n}^{[\mathrm{W}]} ) \nonumber
\\ & \quad = H (\mathbf{X}_{\mathrm{V}_1} , \mathbf{X}_{\mathrm{V}_2}, \mathbf{X}_{l} | \mathbf{Q}_{n}^{[\mathrm{W}]},\mathbf{A}_{n}^{[\mathrm{W}]} ) \nonumber
\\ & \quad = H (\mathbf{X}_{\mathrm{V}_1}  | \mathbf{Q}_{n}^{[\mathrm{W}]},\mathbf{A}_{n}^{[\mathrm{W}]} ) \nonumber
\\ & \quad \quad +  H (\mathbf{X}_{\mathrm{V}_2}, \mathbf{X}_{l} | \mathbf{X}_{\mathrm{V}_1}, \mathbf{Q}_{n}^{[\mathrm{W}]},\mathbf{A}_{n}^{[\mathrm{W}]} ) \label{eq:Xj1
_given_QA_plus_Xj2_Xl_given_Xj1_QA}
\\ & \quad = H (\mathbf{X}_{\mathrm{V}_1} )  + H (\mathbf{X}_{\mathrm{V}_2}, \mathbf{X}_{l} | \mathbf{X}_{\mathrm{V}_1}, \mathbf{Q}_{n}^{[\mathrm{W}']},\mathbf{A}_{n}^{[\mathrm{W}']}) \label{eq:Xj'_plus_Xj''_Xim_given_Xj'_QA}
\\ & \quad = H (\mathbf{X}_{\mathrm{V}_1}  )  + H (\mathbf{X}_{\mathrm{V}_2}, \mathbf{X}_{l} ) \label{eq:Xj'_plus_Xj''_Xim}
\\ & \quad = H (\mathbf{X}_{\mathrm{V}} ) \label{eq:Xj}
\end{align}
where~\eqref{eq:Xj1
_given_QA_plus_Xj2_Xl_given_Xj1_QA} follows from the chain rule;~\eqref{eq:Xj'_plus_Xj''_Xim_given_Xj'_QA} follows from~\eqref{eq:induction_statement} and~\eqref{eq:privacy};~\eqref{eq:Xj'_plus_Xj''_Xim} follows from~\eqref{eq:security} since $(\mathrm{V}_2 \cup \{l\})\cap \mathrm{W}'=\emptyset$; and~\eqref{eq:Xj} follows from the independence of the messages.
This completes the induction step and hence the proof.

Thus, for any ${\mathrm{W}\in \{\mathrm{W}_1,\dots,\mathrm{W}_E\}}$, any ${\mathrm{V} \subseteq \mathrm{W}}$, and any ${n \in [1:N]}$, we have
\begin{equation}
\label{eq:Xj_given_QA=_Xj}
    H ( \mathbf{X}_{\mathrm{V}} | \mathbf{Q}_{n}^{[\mathrm{W}]}, \mathbf{A}_{n}^{[\mathrm{W}]} ) =  H (\mathbf{X}_{\mathrm{V}} ). 
\end{equation}

Finally,  combining~\eqref{eq:Xw_=_Xj_plus_Xj1_given_Xj}, \eqref{eq:Xj1_given_Xj_QA=_Xj1},  and~\eqref{eq:Xj_given_QA=_Xj}, it follows that
\begin{align}
         & H( \mathbf{X}_{\mathrm{U}} |  \mathbf{Q}_{n}^{[\mathrm{W}]}, \mathbf{A}_{n}^{[\mathrm{W}]}) \nonumber
         \\ & \quad = H (\mathbf{X}_{\mathrm{V}} ) +  H (\mathbf{X}_{\mathrm{V}'} ) \nonumber  
         \\ & \quad = H(\mathbf{X}_{\mathrm{U}}  ), \label{eq:X_T}
    \end{align}
where~\eqref{eq:X_T} follows from the independence of the messages. 
This completes the proof.
\end{proof}

We now prove the rate upper bound in Theorem~\ref{thm:SPSSR_capacity}.
Recall from~\eqref{eq:RateDef} that the rate is given by the ratio of $H(\mathbf{X}_{\mathbf{W}})$ to $\sum_{n=1}^N H(\mathbf{A}^{[\mathbf{W}]}_n|\mathbf{Q}^{[\mathbf{W}]}_n)$.
Since the messages are independent and uniformly distributed over
$\mathbb{F}_q^{L}$, we have
\begin{equation}
\label{eq:H_Xw}
H(\mathbf{X}_{\mathbf{W}}) = DL,
\end{equation}
see~\cite{IH2026PSSRarXiv} for details.
Thus, to upper bound the rate, it remains to lower bound
$\sum_{n=1}^N H(\mathbf{A}^{[\mathbf{W}]}_n
|\mathbf{Q}^{[\mathbf{W}]}_n)$.
We do so using an approach similar to those in~\cite{SJ2019Sym,WBU2022} and proceed as follows:
\begin{align}
 & \sum_{n=1}^N H(\mathbf{A}^{[\mathbf{W}]}_n|\mathbf{Q}^{[\mathbf{W}]}_n) \nonumber
 \\  & \quad \geq \sum_{n=1}^{N} H(\mathbf{A}_{n}^{[\mathbf{W}]} | \mathbf{Q}_{n}^{[\mathbf{W}]},\mathbf{W})\label{eq:cond_on_W}
 \\  & \quad  = \sum_{n=1}^{N} H(\mathbf{A}_{n}^{[\mathrm{W}]} | \mathbf{Q}_{n}^{[\mathrm{W}]}) \label{eq:remove_cond_on_W}
 \\  & \quad  = \sum_{n=1}^{N} H(\mathbf{A}_{n}^{[\mathrm{W}]} | \mathbf{Q}_{[1:N]}^{[\mathrm{W}]}) \label{eq: sum_H_An_given_Q1:N}
 \\ & \quad \geq H( \mathbf{A}_{[1:N]}^{[\mathrm{W}]} |  \mathbf{Q}_{[1:N]}^{[\mathrm{W}]})  \label{eq: H_A1:N_given_Q1:N}
 \\ & \quad = H( \mathbf{A}_{[1:N]}^{[\mathrm{W}]},  \mathbf{Q}_{[1:N]}^{[\mathrm{W}]}) - H(  \mathbf{Q}_{[1:N]}^{[\mathrm{W}]}) \label{eq: H_A1:N_Q1:N_minus_H_Q1:N}
 \\ & \quad = H( \mathbf{A}_{[1:N]}^{[\mathrm{W}]},  \mathbf{Q}_{[1:N]}^{[\mathrm{W}]}) - H( \mathbf{A}_{[1:N]}^{[\mathrm{W}]},  \mathbf{Q}_{[1:N]}^{[\mathrm{W}]} |\mathbf{X}_{\mathrm{W}} ) \nonumber 
 \\ & \quad \quad + H( \mathbf{A}_{[1:N]}^{[\mathrm{W}]},  \mathbf{Q}_{[1:N]}^{[\mathrm{W}]} |\mathbf{X}_{\mathrm{W}} )- H(  \mathbf{Q}_{[1:N]}^{[\mathrm{W}]}) \nonumber 
 \\ & \quad = I( \mathbf{X}_{\mathrm{W}}; \mathbf{A}_{[1:N]}^{[\mathrm{W}]},  \mathbf{Q}_{[1:N]}^{[\mathrm{W}]}) \nonumber \\ 
 & \quad \quad + H( \mathbf{A}_{[1:N]}^{[\mathrm{W}]},  \mathbf{Q}_{[1:N]}^{[\mathrm{W}]} |\mathbf{X}_{\mathrm{W}} )- H(  \mathbf{Q}_{[1:N]}^{[\mathrm{W}]})  \label{eq: I_Xw_A1:N_Q1:N_plus_H_A1:N_Q1:N_given_Xw_minus_H_Q1:N}
 \\ & \quad = H( \mathbf{X}_{\mathrm{W}}) - H( \mathbf{X}_{\mathrm{W}} | \mathbf{A}_{[1:N]}^{[\mathrm{W}]},  \mathbf{Q}_{[1:N]}^{[\mathrm{W}]}) \nonumber 
 \\ & \quad \quad + H( \mathbf{A}_{[1:N]}^{[\mathrm{W}]},  \mathbf{Q}_{[1:N]}^{[\mathrm{W}]} |\mathbf{X}_{\mathrm{W}} )- H(  \mathbf{Q}_{[1:N]}^{[\mathrm{W}]}) \label{eq: H_Xw_minus_H_Xw_given_A1:N_Q1:N_plus_H_A1:N_Q1:N_given_Xw_minus_H_Q1:N}
 \\ & \quad = DL + H( \mathbf{A}_{[1:N]}^{[\mathrm{W}]},  \mathbf{Q}_{[1:N]}^{[\mathrm{W}]} |\mathbf{X}_{\mathrm{W}} )- H(  \mathbf{Q}_{[1:N]}^{[\mathrm{W}]}) \label{eq: DL_plus_H_A1:N_Q1:N_given_Xw_minus_H_Q1:N}
 \\ & \quad = DL + H(  \mathbf{Q}_{[1:N]}^{[\mathrm{W}]} |\mathbf{X}_{\mathrm{W}} ) \nonumber \\ 
 & \quad \quad + H( \mathbf{A}_{[1:N]}^{[\mathrm{W}]}|  \mathbf{Q}_{[1:N]}^{[\mathrm{W}]} ,\mathbf{X}_{\mathrm{W}} )- H(  \mathbf{Q}_{[1:N]}^{[\mathrm{W}]}) \label{eq: DL_plus_H_Q1:N_given_Xw_plus_H_A1:N_given_Q1:N_Xw_minus_H_Q1:N}
 \\ & \quad = DL + H( \mathbf{A}_{[1:N]}^{[\mathrm{W}]}|  \mathbf{Q}_{[1:N]}^{[\mathrm{W}]} ,\mathbf{X}_{\mathrm{W}} ) \label{eq: DL_plus_H_A1:N_given_Q1:N_Xw}
 \\ & \quad \geq DL +\frac{1}{N} \sum_{n=1}^N H( \mathbf{A}_{n}^{[\mathrm{W}]}|  \mathbf{Q}_{[1:N]}^{[\mathrm{W}]} ,\mathbf{X}_{\mathrm{W}} ) \label{eq: DL_plus_1/N_sum_H_An_given_Q1:N_Xw}
 \\ & \quad = DL +\frac{1}{N} \sum_{n=1}^N H( \mathbf{A}_{n}^{[\mathrm{W}]}|  \mathbf{Q}_{n}^{[\mathrm{W}]} ,\mathbf{X}_{\mathrm{W}} ) \label{eq: DL_plus_1/N_sum_H_An_given_Qn_Xw}
 \\ & \quad = DL +\frac{1}{N} \sum_{n=1}^N H( \mathbf{A}_{n}^{[\mathrm{W}]}|  \mathbf{Q}_{n}^{[\mathrm{W}]} ) \label{eq: DL_plus_1/N_sum_H_An_given_Qn}
\end{align}
where~\eqref{eq:cond_on_W} follows since conditioning cannot increase entropy; 
\eqref{eq:remove_cond_on_W} follows from Lemma~\ref{lemma: symmetry};
\eqref{eq: sum_H_An_given_Q1:N} 
and~\eqref{eq: DL_plus_1/N_sum_H_An_given_Qn_Xw} follow from Lemma~\ref{lemma: given_Q1:N_=_given_Qn}; 
\eqref{eq: H_A1:N_given_Q1:N} follows from the subadditivity of entropy; 
\eqref{eq: H_A1:N_Q1:N_minus_H_Q1:N} and~\eqref{eq: DL_plus_H_Q1:N_given_Xw_plus_H_A1:N_given_Q1:N_Xw_minus_H_Q1:N} follow from the chain rule; 
\eqref{eq: I_Xw_A1:N_Q1:N_plus_H_A1:N_Q1:N_given_Xw_minus_H_Q1:N} and~\eqref{eq: H_Xw_minus_H_Xw_given_A1:N_Q1:N_plus_H_A1:N_Q1:N_given_Xw_minus_H_Q1:N} follow from the definition of mutual information; 
\eqref{eq: DL_plus_H_A1:N_Q1:N_given_Xw_minus_H_Q1:N} follows from~\eqref{eq:correctness}; 
\eqref{eq: DL_plus_H_A1:N_given_Q1:N_Xw} follows from~\eqref{eq:Q_indep_X_S}; 
\eqref{eq: DL_plus_1/N_sum_H_An_given_Q1:N_Xw} follows from the monotonicity of entropy; 
and~\eqref{eq: DL_plus_1/N_sum_H_An_given_Qn} follows from Lemma~\ref{lemma: dropping_undesired_subset_from_condition}.

Rearranging~\eqref{eq: DL_plus_1/N_sum_H_An_given_Qn} yields
\begin{equation}
\label{eq:NDL/N-1}
    \sum_{n=1}^N H(\mathbf{A}^{[\mathbf{W}]}_n|\mathbf{Q}^{[\mathbf{W}]}_n) \geq \frac{NDL}{N-1}.
\end{equation}
Combining~\eqref{eq:H_Xw} and~\eqref{eq:NDL/N-1} yields the rate upper bound ${1-1/N}$, matching $R$ in~\eqref{eq:SPSSR_cap}.

\subsection{Lower Bounding the Shared-Randomness Ratio}
\label{sec:M/L_conv_proof_SPSSR}

Next, we derive a lower bound on the required number of shared secret keys as follows: 
\begin{align}
M & = H(\mathbf{S}) 
 \nonumber \\ & = \frac{1}{N} \sum_{n=1}^N 
 H(\mathbf{S} | \mathbf{Q}^{[\mathrm{W}]}_n, \mathbf{X}_{[1:K]})
 \label{eq:H_S_given_Q_X}
 \\ & \geq \frac{1}{N} \sum_{n=1}^N  
 \Big[
 H(\mathbf{S} | \mathbf{Q}^{[\mathrm{W}]}_n, \mathbf{X}_{[1:K]})
 \nonumber \\ & \qquad\qquad
 - H(\mathbf{S} | \mathbf{A}^{[\mathrm{W}]}_n,
 \mathbf{Q}^{[\mathrm{W}]}_n, \mathbf{X}_{[1:K]})
 \Big]
 \label{eq:H_S_given_Q_X_minus_H_S_given_A_Q_X}
 \\ & = \frac{1}{N} \sum_{n=1}^N 
 I(\mathbf{S}; \mathbf{A}^{[\mathrm{W}]}_n
 | \mathbf{Q}^{[\mathrm{W}]}_n, \mathbf{X}_{[1:K]})
 \label{eq:I_SandA_given_Q_X}
 \\ & = \frac{1}{N} \sum_{n=1}^N  
 \Big[
 H(\mathbf{A}^{[\mathrm{W}]}_n
 | \mathbf{Q}^{[\mathrm{W}]}_n, \mathbf{X}_{[1:K]})
 \nonumber \\ & \qquad\qquad
 - H(\mathbf{A}^{[\mathrm{W}]}_n
 | \mathbf{S}, \mathbf{Q}^{[\mathrm{W}]}_n, \mathbf{X}_{[1:K]})
 \Big]
 \label{eq:H_A_given_Q_X_minus_H_A_given_S_Q_X}
 \\ & = \frac{1}{N} \sum_{n=1}^N  
 H(\mathbf{A}^{[\mathrm{W}]}_n
 | \mathbf{Q}^{[\mathrm{W}]}_n, \mathbf{X}_{[1:K]})
 \label{eq:H_A_given_Q_X}
 \\ & = \frac{1}{N} \sum_{n=1}^N  
 \Big[
 H(\mathbf{A}^{[\mathrm{W}]}_n
 | \mathbf{Q}^{[\mathrm{W}]}_n, \mathbf{X}_{\mathrm{W}})
 \nonumber \\ & \qquad\qquad
 - I(\mathbf{X}_{\overline{\mathrm{W}}};
 \mathbf{A}^{[\mathrm{W}]}_n
 | \mathbf{Q}^{[\mathrm{W}]}_n, \mathbf{X}_{\mathrm{W}})
 \Big]
 \label{eq:H_A_given_Q_Xw_minus_I_XinterfandA_givenQ_Xw}
 \\ & = \frac{1}{N} \sum_{n=1}^N  
 \Big[
 H(\mathbf{A}^{[\mathrm{W}]}_n
 | \mathbf{Q}^{[\mathrm{W}]}_n, \mathbf{X}_{\mathrm{W}})
 \nonumber \\ & \qquad\qquad
 - H(\mathbf{X}_{\overline{\mathrm{W}}}
 | \mathbf{Q}^{[\mathrm{W}]}_n, \mathbf{X}_{\mathrm{W}})
 \nonumber \\ & \qquad\qquad
 + H(\mathbf{X}_{\overline{\mathrm{W}}}
 | \mathbf{A}^{[\mathrm{W}]}_n,
 \mathbf{Q}^{[\mathrm{W}]}_n, \mathbf{X}_{\mathrm{W}})
 \Big]
 \label{eq:H_A_given_Q_Xw_minus_H_Xinterf_givenQ_Xw_plus_H_Xinterf_givenAQXw}
 \\ & = \frac{1}{N} \sum_{n=1}^N  
 \Big[
 H(\mathbf{A}^{[\mathrm{W}]}_n
 | \mathbf{Q}^{[\mathrm{W}]}_n, \mathbf{X}_{\mathrm{W}})
 \nonumber \\ & \qquad\qquad
 - H(\mathbf{X}_{\overline{\mathrm{W}}})
 + H(\mathbf{X}_{\overline{\mathrm{W}}})
 \Big]
 \label{eq:H_A_given_Q_Xw_minus_H_Xinterf_plus_H_Xinterf}
 \\ & = \frac{1}{N} \sum_{n=1}^N  
 H(\mathbf{A}^{[\mathrm{W}]}_n
 | \mathbf{Q}^{[\mathrm{W}]}_n)
 \label{eq:H_A_given_Q}
 \\ & \geq \frac{DL}{N-1}
 \label{eq:DL/N-1}
\end{align}
where~\eqref{eq:H_S_given_Q_X} follows from~\eqref{eq:Q_indep_X_S};~\eqref{eq:H_S_given_Q_X_minus_H_S_given_A_Q_X} follows since entropy is non-negative;~\eqref{eq:I_SandA_given_Q_X},~\eqref{eq:H_A_given_Q_X_minus_H_A_given_S_Q_X},~\eqref{eq:H_A_given_Q_Xw_minus_I_XinterfandA_givenQ_Xw} and~\eqref{eq:H_A_given_Q_Xw_minus_H_Xinterf_givenQ_Xw_plus_H_Xinterf_givenAQXw} follow from the definition of mutual information;~\eqref{eq:H_A_given_Q_X} follows from~\eqref{eq:A_func_Q_X_S};~\eqref{eq:H_A_given_Q_Xw_minus_H_Xinterf_plus_H_Xinterf} follows from the independence of the messages and queries and~\eqref{eq:security};~\eqref{eq:H_A_given_Q} follows from Lemma~\ref{lemma: dropping_undesired_subset_from_condition}; and~\eqref{eq:DL/N-1} follows from~\eqref{eq:NDL/N-1}.

Rearranging~\eqref{eq:DL/N-1} yields 
\[\frac{M}{L}\geq \frac{D}{N-1}.\] 
Thus, the minimum shared-randomness ratio required to achieve rate $R$ is ${D/(N-1)}$, which matches $M/L$ in~\eqref{eq:kappa_SPSSR}.

\subsection{Lower Bounding the Subpacketization Level}
\label{sec:L_conv_proof_SPSSR}

Any balanced linear SPSSR scheme achieving rate $R$ retrieves a total of $DL/R$ linear combinations of message subpackets and shared secret keys, evenly distributed across the $N$ servers. 
Hence, ${DL}/{(NR)}$ must be an integer. 
Using~\eqref{eq:SPSSR_cap}, 
\[
\frac{DL}{NR} = \frac{DL}{N-1},
\]
so the smallest feasible subpacketization level is ${(N-1)/\gcd(D,N-1)}$, matching $L$ in~\eqref{eq:L*_SPSSR}.

\section{Achievability Proofs}
\label{sec:achiev_proof_SPSSR}

In this section, we present a single balanced ${\{0,1\}}$-linear SPSSR scheme that applies uniformly to every demand family. 
The scheme achieves the rate $R$ given in Theorem~\ref{thm:SPSSR_capacity} and requires the shared-randomness ratio $M/L$ and subpacketization level $L$ specified in Theorem~\ref{thm:SPSSR_L*_bounds}.

\subsection{Achievable Scheme}

Each message is divided into ${L=(N-1)/G}$ subpackets, where ${G\coloneqq\gcd(D,N-1)}$. 
For each ${i\in[1:K]}$, we denote the $l$th subpacket of message
$\mathrm{X}_i$ by $\mathrm{X}_{i,l}$.
 
The user first partitions
$\mathrm{W}=\{i_1,\dots,i_D\}$ into ${M=D/G}$ groups of size $G$, with the $m$th group given by
\[
\{i_{(m-1)G+1},\dots,i_{mG}\}, \qquad m\in[1:M].
\]
Then, the user generates an ${M\times KL}$ matrix whose entries are drawn
independently and uniformly from ${\{0,1\}}$, and sends this matrix as the query to the first server, i.e.,
\begin{equation*}
    \mathrm{Q}^{[\mathrm{W}]}_{1} =
    \begin{bmatrix}
        h_{1,1}^{(1)} & \cdots & h_{1,L}^{(1)} & \cdots & h_{K,1}^{(1)} & \cdots & h_{K,L}^{(1)} \\
        \vdots & \vdots & \vdots & \vdots & \vdots & \vdots & \vdots\\
        h_{1,1}^{(M)} & \cdots & h_{1,L}^{(M)} & \cdots & h_{K,1}^{(M)} & \cdots & h_{K,L}^{(M)}
    \end{bmatrix},
\end{equation*}
where ${h_{i,l}^{(m)} \in \{0,1\}}$ for all ${m \in [1:M]}$, ${i \in [1:K]}$, and ${l \in [1:L]}$.

The queries for the remaining servers are then generated from
$\mathrm{Q}^{[\mathrm{W}]}_{1}$ by flipping one entry in each row.
For each ${n\in [2:N]}$, define
\[
g_n\coloneqq \left\lceil \frac{n-1}{L}\right\rceil,
\qquad
l_n\coloneqq (n-1)-(g_n-1)L .
\]
Then each ${n\in[2:N]}$ can be written uniquely as
\[
n=1+(g_n-1)L+l_n,
\]
which establishes a one-to-one correspondence between
${n\in[2:N]}$ and ${(g_n,l_n)\in[1:G]\times[1:L]}$.
In particular, ${n=2}$ corresponds to ${(g_2,l_2)=(1,1)}$, while
${n=N}$ corresponds to ${(g_N,l_N)=(G,L)}$.

For each server ${n\in [2:N]}$, the query $\mathrm{Q}^{[\mathrm{W}]}_{n}$ is an ${M \times KL}$ matrix obtained from $\mathrm{Q}^{[\mathrm{W}]}_{1}$ as follows. 
For each ${m\in[1:M]}$, the user flips, in the $m$th row, the entry
$h_{i,l}^{(m)}$ with ${i=i_{(m-1)G+g_n}}$ and ${l=l_n}$.
That is, if this entry is equal to $1$ in $\mathrm{Q}^{[\mathrm{W}]}_{1}$,
it is set to $0$ in $\mathrm{Q}^{[\mathrm{W}]}_{n}$, and if it is equal to
$0$, it is set to $1$.
We denote the flipped version of $h_{i,l}^{(m)}$ by
$\overline{h}_{i,l}^{(m)}$, for all ${m \in [1:M]}$, ${i \in [1:K]}$, and ${l \in [1:L]}$.

Upon receiving $\mathrm{Q}^{[\mathrm{W}]}_{n}$, each server ${n\in[1:N]}$ returns $M$ masked linear combinations of message subpackets.
For each ${m\in[1:M]}$, the coefficients of the $m$th linear combination are specified by the $m$th row of $\mathrm{Q}^{[\mathrm{W}]}_{n}$, and the corresponding shared secret key $\mathrm{S}_m$ is added as a mask.
For example, the answer returned by the first server is
\begin{equation*}
    \mathrm{A}^{[\mathrm{W}]}_{1} =
    \begin{bmatrix}
        \sum_{i=1}^{K} \sum_{l=1}^{L} h_{i,l}^{(1)} \mathrm{X}_{i,l} + \mathrm{S}_1 \\
        \vdots \\
        \sum_{i=1}^{K} \sum_{l=1}^{L} h_{i,l}^{(M)} \mathrm{X}_{i,l} + \mathrm{S}_M \\
    \end{bmatrix}.
\end{equation*}

\subsection{Proof of Optimality}

Since the user retrieves $M$ combinations from each server, each of the size of one message subpacket, the retrieval rate of the scheme is
${DL/(NM)=1-1/N}$, matching~\eqref{eq:SPSSR_cap}.
The scheme requires shared-randomness ratio ${M/L=D/(N-1)}$, which matches~\eqref{eq:kappa_SPSSR}.
The scheme uses subpacketization level ${L=(N-1)/G}$, matching~\eqref{eq:L*_SPSSR}.
This proves the optimality of the scheme in terms of retrieval rate, shared-randomness ratio, and subpacketization level. 

\subsection{Proof of Correctness}

For each ${n\in[2:N]}$ and each ${m\in[1:M]}$, the entry $h_{i,l}^{(m)}$ in $\mathrm{Q}^{[\mathrm{W}]}_{1}$, with ${i=i_{(m-1)G+g_n}}$ and ${l=l_n}$, is the combining coefficient of the $l$th subpacket of the demand message indexed by $i$ in the $m$th combination retrieved from Server~1.
The corresponding coefficient in the $m$th combination retrieved from each server $n$ is flipped to $\overline{h}_{i,l}^{(m)}$, while all other coefficients remain the same as in Server~1.
Additionally, the same secret key $\mathrm{S}_m$ is used in the $m$th combination retrieved from all servers. Therefore, subtracting $\mathrm{A}^{[\mathrm{W}]}_{1}$ from $\mathrm{A}^{[\mathrm{W}]}_{n}$ cancels all unchanged terms and recovers $M$ subpackets from $M$ distinct demand messages, one from each component of the vector $\mathrm{A}^{[\mathrm{W}]}_{n}-\mathrm{A}^{[\mathrm{W}]}_{1}$. Considering all ${n\in[2:N]}$, the user recovers a total of ${M(N-1)}$ demand subpackets, which equals the total number of demand subpackets $DL$.

\subsection{Proof of Privacy}

The scheme satisfies the privacy condition since, for every demand set, each server observes a query matrix of the same size with entries that appear independent and uniformly random from that server's perspective. 
Therefore, for each ${n \in [1:N]}$, the query $\mathrm{Q}^{[\mathrm{W}]}_{n}$ reveals no information about the demand index set to server $n$.

\subsection{Proof of Security}

The scheme also satisfies the security condition.
At each server, every linear combination of message subpackets is masked by a shared secret key, so the answer from any single server reveals no information about the messages.
After combining answers from multiple servers, the user recovers only demand subpackets, while all interference-message terms cancel out.
Thus, the user obtains no information about the interference messages.

\section{An Illustrative Example}

In this section, we present an illustrative example of the proposed SPSSR scheme.

Consider a set of $K=6$ messages stored on $N=3$ servers.
We index the messages by $1,2,3,4,5,6$, and for notational convenience, we denote them by $a,b,c,d,e,f$, respectively.

Suppose a user wishes to retrieve ${D=4}$ messages, with the demand index set $\mathrm{W}$ belonging to the full collection of all $4$-subsets of $[1:6]$, denoted by $\mathrm{W}_1,\dots,\mathrm{W}_{15}$. The same scheme applies to any subcollection $\mathrm{W}_1,\dots,\mathrm{W}_{E}$ satisfying $\left|\bigcup_{j=1}^{E}\mathrm{W}_j\right|=6$ and $\left|\bigcap_{j=1}^{E}\mathrm{W}_j\right|=0$,
and achieves the same rate, subpacketization level, and shared-randomness ratio.

In this example, we have ${G=\gcd(D,N-1)=2}$, ${L=(N-1)/G=1}$, and ${M=D/G=2}$. 
Thus, the demand index set $\mathrm{W}$ is partitioned into two groups of two indices each. 
For instance, when $\mathrm{W}=\{1,2,3,4\}$, the first group consists of $i_1=1$ and $i_2=2$ and the second group consists of $i_3=3$ and $i_4=4$.

Since each message consists of a single subpacket, i.e., ${L=1}$, we write ${h}_{i}^{(m)}$ instead of ${h}_{i,l}^{(m)}$ for ease of notation, where ${m\in[1:2]}$ and ${i\in[1:6]}$.
The query to the first server is
\[
\mathrm{Q}^{[\mathrm{W}]}_{1}
=
\begin{bmatrix}
h_{1}^{(1)} & h_{2}^{(1)} & h_{3}^{(1)}
& h_{4}^{(1)} & h_{5}^{(1)} & h_{6}^{(1)}
\\
h_{1}^{(2)} & h_{2}^{(2)} & h_{3}^{(2)}
& h_{4}^{(2)} & h_{5}^{(2)} & h_{6}^{(2)}
\end{bmatrix},
\]
where the entries ${h}_{i}^{(m)}$ are generated independently and
uniformly at random from ${\{0,1\}}$.

By construction, ${n=2}$ corresponds to ${(g_2,l_2)=(1,1)}$, and ${n=3}$ corresponds to ${(g_3,l_3)=(G,L)=(2,1)}$. 
Hence, $\mathrm{Q}^{[\mathrm{W}]}_{2}$ is obtained from $\mathrm{Q}^{[\mathrm{W}]}_{1}$ by flipping ${h_{i_1}^{(1)}}$ in the first row and ${h_{i_3}^{(2)}}$ in the second row.
For instance, if ${\mathrm{W}=\{1,2,3,4\}}$, then ${i_1=1}$ and ${i_3=3}$, so
\[
\mathrm{Q}^{[\mathrm{W}]}_{2}
=
\begin{bmatrix}
\overline{h}_{1}^{(1)} & h_{2}^{(1)} & h_{3}^{(1)}
& h_{4}^{(1)} & h_{5}^{(1)} & h_{6}^{(1)}
\\
h_{1}^{(2)} & h_{2}^{(2)} & \overline{h}_{3}^{(2)}
& h_{4}^{(2)} & h_{5}^{(2)} & h_{6}^{(2)}
\end{bmatrix}.
\]

Similarly, $\mathrm{Q}^{[\mathrm{W}]}_{3}$ is obtained from $\mathrm{Q}^{[\mathrm{W}]}_{1}$ by flipping ${h_{i_2}^{(1)}}$ in the first row and ${h_{i_4}^{(2)}}$ in the second row.
Thus, if ${\mathrm{W}=\{1,2,3,4\}}$, then ${i_2=2}$ and ${i_4=4}$, so
\[
\mathrm{Q}^{[\mathrm{W}]}_{3}
=
\begin{bmatrix}
h_{1}^{(1)} & \overline{h}_{2}^{(1)} & h_{3}^{(1)}
& h_{4}^{(1)} & h_{5}^{(1)} & h_{6}^{(1)}
\\
h_{1}^{(2)} & h_{2}^{(2)} & h_{3}^{(2)}
& \overline{h}_{4}^{(2)} & h_{5}^{(2)}
& h_{6}^{(2)}
\end{bmatrix}.
\]

Table~\ref{tab:SPSSR_W1} shows the answers returned by the three servers for the case ${\mathrm{W}=\{1,2,3,4\}}$, where the demand messages are $a$, $b$, $c$, and $d$.
The answer tables for the other candidate demand index sets are obtained by relabeling the messages.

{\renewcommand{\arraystretch}{2}
\begin{table*}[t]
    \centering
    \caption{Answer table for the case $\mathrm{W}=\{1,2,3,4\}$
(demand messages: $a$, $b$, $c$, and $d$)}
 \label{tab:SPSSR_W1}
    \scalebox{1.05}{
    \begin{tabular}{|c|c|c|}
    \hline
     Server 1 & Server 2 & Server 3 \\
    \hline
    \hline
    ${h}_{1}^{(1)}a+{h}_{2}^{(1)}b+{h}_{3}^{(1)}c +{h}_{4}^{(1)}d+{h}_{5}^{(1)}e+{h}_{6}^{(1)}f+\mathrm{S}_1$ & $\overline{h}_{1}^{(1)} a+{h}_{2}^{(1)}b+{h}_{3}^{(1)}c +{h}_{4}^{(1)}d+{h}_{5}^{(1)}e+{h}_{6}^{(1)}f+\mathrm{S}_1$ & ${h}_{1}^{(1)}a+\overline{{h}}_{2}^{(1)}b+{h}_{3}^{(1)}c +{h}_{4}^{(1)}d+{h}_{5}^{(1)}e+{h}_{6}^{(1)}f+\mathrm{S}_1$ \\

    \hline
${h}_{1}^{(2)}a+{h}_{2}^{(2)}b+{h}_{3}^{(2)}c +{h}_{4}^{(2
)}d+{h}_{5}^{(2)}e+{h}_{6}^{(2)}f+\mathrm{S}_2$ & ${h}_{1}^{(2)}a+{h}_{2}^{(2)}b+\overline{{h}}_{3}^{(2)}c +{h}_{4}^{(2)}d+{h}_{5}^{(2)}e+{h}_{6}^{(2)}f+\mathrm{S}_2$ & ${h}_{1}^{(2)}a+{h}_{2}^{(2)}b+{h}_{3}^{(2)}c +\overline{{h}}_{4}^{(2)}d+{h}_{5}^{(2)}e+{h}_{6}^{(2)}f+\mathrm{S}_2$ \\

    \hline
    \end{tabular}
    }
\end{table*}
}

Combining the first combinations retrieved from Servers~1 and~2 enables the user to recover message $a$.
Indeed, the shared secret key $\mathrm{S}_1$ cancels, and all terms corresponding to messages $b$, $c$, $d$, $e$, and $f$ cancel because their coefficients are identical in the two combinations; only the coefficient of message $a$ differs.
Similarly, combining the first combinations retrieved from Servers~1 and~3 enables the user to recover message $b$.
Applying the same argument to the second retrieved combinations, the user recovers message $c$ from Servers~1 and~2 and message $d$ from Servers~1 and~3.
This proves correctness for the case $\mathrm{W}=\{1,2,3,4\}$.

Since each server observes the same query structure for every demand set, and since the entries ${h}_{i}^{(m)}$, for ${m\in[1:2]}$ and ${i\in[1:6]}$, appear independent and uniformly random from that server's perspective, the privacy condition is satisfied.

The answer from any single server reveals no information about the messages, since each retrieved combination is masked by either $\mathrm{S}_1$ or $\mathrm{S}_2$. 
Even after combining answers from multiple servers, the user recovers only the demand messages $a$, $b$, $c$, and $d$, and obtains no information about the interference messages $e$ and $f$. 
Thus, the security condition is satisfied.

The user retrieves two linear combinations from each server, each of
the same size as one message subpacket, while the demand consists of ${D=4}$ messages, each with one subpacket. 
Thus, the retrieval rate is
\[R = \frac{4\cdot 1}{3\cdot 2} =\frac{2}{3},\] 
which matches~\eqref{eq:SPSSR_cap}. Moreover, the shared-randomness ratio is ${M/L=2}$, matching~\eqref{eq:kappa_SPSSR}, and the subpacketization level is ${L=1}$, which matches~\eqref{eq:L*_SPSSR}.

Finally, we compare our scheme with SMPIR schemes under the same parameters
${N=3}$, ${K=6}$, and ${D=4}$.
In this setting, the SMPIR scheme of~\cite{WBU2022} achieves the same rate as our scheme but requires subpacketization level ${36}$.
The same rate can also be achieved by applying the SPIR scheme of~\cite{SJ2019Sym} successively ${D=4}$ times, retrieving one demand message at a time; this scheme requires subpacketization level ${N-1=2}$.
In contrast, our scheme requires only ${L=1}$.
All three schemes have the same shared-randomness ratio ${M/L=2}$.

\section{Open Problems and Future Directions}

Several important questions remain open regarding the optimal tradeoffs among retrieval rate, shared-randomness ratio, and subpacketization level for balanced ${\{0,1\}}$-linear schemes. 
In particular, what are the minimum shared-randomness ratio and the minimum subpacketization level  required to achieve a retrieval rate above a given threshold, and can these minima be achieved simultaneously? 
Similarly, under given upper bounds on shared-randomness ratio and subpacketization level, what is the highest achievable retrieval rate?

In addition to these questions, several broader directions remain for future work. 

The minimum subpacketization level established here applies only to balanced schemes, in which the retrieved combinations are evenly distributed across servers. 
Extending the analysis to asymmetric schemes and determining whether they can achieve smaller subpacketization levels is an interesting direction for future work.

This work studies prior-agnostic SPSSR with equal-length messages, so the results apply to any full-support demand prior but do not cover unequal message lengths. 
A related direction is prior-aware SPSSR with heterogeneous message lengths, motivated by semantic PIR~\cite{VBU2022}.

Another direction is to study SPSSR under relaxed privacy and security requirements. 
In particular, one may protect each demand message separately, rather than protecting the demand set as a whole; in related settings, this relaxation is known to yield more efficient schemes~\cite{HKRS2019,HS2021,HS2022LinCap}. 
Moreover, one may require that the user learn nothing about any subset of interference messages up to a given size, rather than about the full set of interference messages.
The goal would be to characterize how these relaxations affect the optimal retrieval rate, shared-randomness ratio, and subpacketization level.

SPSSR can also be extended to settings where the user has side information about some messages, or functions of them. 
Since such side information can improve the efficiency of PIR and MPIR schemes in various settings~\cite{KGHERS2017,KGHERS2017No0,HKGRS2018,SSM2018,LG2018,HKS2018,HKS2019Journal,HKS2019,KKHS12019,KKHS22019,WHS2024,EH2024}, a natural direction for future work is to investigate whether similar gains are possible in SPSSR.

Finally, this work assumes honest and non-colluding servers, where each server observes only its own query and follows the protocol. 
A possible future direction is to study SPSSR with colluding or adversarial servers, where servers may combine their observations or deviate from the protocol, as considered for classical PIR in~\cite{SJ2018Colluding,BU2019Colluding,ZX2019,WS2017}.

\balance

\bibliographystyle{IEEEtran}
\bibliography{PIR_PC_Refs}

\end{document}